\documentclass[review]{elsarticle}

\usepackage{amsmath,
            amsthm,
            endfloat,
            mathrsfs,
            multirow}

\usepackage{tikz}
\usetikzlibrary{arrows}
\tikzstyle{vertex} = [circle,fill=black!0,minimum size=4pt,inner sep=1pt]
\tikzstyle{every node}=[transform shape]

\usepackage[section]{algorithm}
\usepackage{algorithmic}
\newcommand{\algorithmicinputdata}{\textbf{Input: }}
\newcommand{\algorithmicresult}{\textbf{Output: }}

\newtheorem{theorem}{Theorem}
\newtheorem{lemma}{Lemma}
\newdefinition{definition}{Definition}

\newcommand{\eg}{e.g.\ }
\newcommand{\ie}{i.e.\ }
\newcommand{\resp}{resp.\ }

\title{Polynomial-time sortable stacks of burnt pancakes}
\author{Anthony Labarre}
\ead{labarre.anthony@gmail.com}
\author{Josef Cibulka\fnref{fn1}}
\ead{cibulka@kam.mff.cuni.cz}
\address{Department of Applied Mathematics, Charles University, Malostransk\'e n\'am.~25, 118~00 Prague, Czech Republic. }
\fntext[fn1]{Josef Cibulka was supported by the project 1M0545
of the Ministry of Education of the Czech Republic
and by the Czech Science Foundation under the contract no.\ 201/09/H057.}

\begin{document}

\begin{abstract}
Pancake flipping, a famous open problem in computer science, can be formalised as the problem of sorting a permutation of positive integers using as few \emph{prefix reversals} as possible. In that context, a prefix reversal of length $k$ reverses the order of the first $k$ elements of the permutation. The \emph{burnt} variant of pancake flipping involves permutations of \emph{signed} integers, and reversals in that case not only reverse the order of elements but also invert their signs. Although three decades have now passed since the first works on these problems, neither their computational complexity nor the maximal number of prefix reversals needed to sort a permutation is yet known. In this work, we prove a new lower bound for sorting burnt pancakes, and show that an important class of permutations, known as ``simple permutations'', can be optimally sorted in polynomial time.
\end{abstract}


\maketitle


\section{Introduction}

The pancake flipping problem~\cite{dweighter-elementary} consists in finding the minimum number of flips required to rearrange a stack of pancakes that all come in different sizes so that the smallest ends up on top and the largest lies at the bottom. 
The problem can be more formally stated as follows: given an ordering of $\{1,2,\ldots,n\}$, what is the minimum number of \emph{prefix reversals} required to sort these numbers in increasing order (where a prefix reversal of length $k$ reverses the order of the first $k$ elements)? A variant of the problem, known as the \emph{burnt pancake flipping problem}, is concerned with rearranging stacks of pancakes that are burnt on one side, in such a way that the pancakes not only end up rearranged in increasing sizes but also with their burnt side down. Again, a more formal description of the problem is to sort orderings of $\{\pm 1, \pm 2, \ldots, \pm n\}$ using as few prefix \emph{signed} reversals (which not only reverse the order of the first $k$ elements but also invert their signs) as possible.

\citet{gates-bounds} and \citet{gyori-stack} proved the first results on pancake flipping three decades ago, focusing on the number of prefix reversals needed in the worst case (\ie the maximum number of steps required to sort a stack of size $n$), and a tremendous amount of work (see next paragraph for more details) has since been devoted to the study of both variants of the problem. 
However, the computational complexity of the sorting problems or merely computing the minimum number of required steps remains open, as well as that of determining the maximum number of prefix (signed) reversals needed to sort a permutation. The best known approximation ratio is $2$, both in the unsigned case (see \citet{fischer-approx}) 
and in the signed case (see \citet{cohen-burnt}, according to \citet{fischer-approx}). 
A very interesting and original solution to the burnt pancake flipping problem was recently proposed by \citet{haynes-engineering}, who use bacteria to represent permutations, which eventually become antibiotic resistant when they are sorted. However, their model obviously does not yield any combinatorial or algorithmic insight on pancake flipping, and seems therefore of little theoretical help.

Although pancake flipping was introduced as a game, it is worth noting that it has since found applications in parallel computing, leading to the famous ``(burnt) pancake network'' topology which is the Cayley graph of a permutation group generated by prefix (signed) reversals (see \citet{lakshmivarahan-symmetry} for a thorough survey on the use of Cayley graphs as interconnection networks). Another major application of pancake flipping, which has received considerable attention, is in the field of computational biology, where permutations model genomes and reversals correspond to actual mutations by which those genomes evolve. In that setting, reversals are no longer restricted to the prefix of the permutation: they can act on any of its segments. Interestingly enough, more is known about these seemingly more challenging versions than on the original pancake flipping problems: \citet{caprara-sorting} proved that sorting unsigned permutations by arbitrary reversals was NP-hard, but quite surprisingly, \citet{hannenhalli-transforming} proved that the signed version of this problem could be solved in polynomial time. For an extensive survey of the mathematical aspects of genome comparisons by means of large-scale mutations, known as \emph{genome rearrangements}, see for instance \citet{fertin-combinatorics}.

In this paper, we use ideas introduced in a previous paper~\cite{labarre-edit} to prove a new tight lower bound on the minimum number of prefix signed reversals required to sort a permutation, which is also referred to as their \emph{distance}. We also examine an important class of signed permutations, known as ``simple permutations'', which proved crucial in solving the problem of sorting permutations by unrestricted signed reversals in polynomial time (see \citet{hannenhalli-transforming}), and give a polynomial-time algorithm for sorting these permutations optimally, as well as a formula for computing their distance in polynomial time. 


\section{Background}

\subsection{Permutations}

Let us start with a quick reminder of basic notions on permutations (for more details, see \eg \citet{bjorner-combinatorics} and \citet{wielandt-finite}). 

\begin{definition}
A \emph{permutation} of $\{1, 2, \ldots, n\}$ is a bijective application of $\{1, 2$, $\ldots$, $n\}$ onto itself.
\end{definition}

The \emph{symmetric group} $S_n$ is the set of all permutations of $\{1, 2, \ldots, n\}$, together with the usual function composition $\circ$, applied from right to left. We use lower case Greek letters to denote permutations, typically $\pi=\langle\pi_1\ \pi_2\ \cdots\ \pi_n\rangle$, with $\pi_i=\pi(i)$, and in particular write the \emph{identity permutation} as $\iota=\langle 1\ 2\ \cdots\ n\rangle$.

\begin{definition}
The \emph{graph} $\Gamma(\pi)$ of a permutation $\pi$ has vertex set $\{1,2,\ldots$, $n\}$, and  contains an arc $(i,j)$ whenever $\pi_i=j$.
\end{definition}

As Figure~\ref{fig:exemple-graph-of-a-permutation} shows, $\Gamma(\pi)$ decomposes in a unique way into disjoint cycles (up to the ordering of cycles and of elements within each cycle), leading to another notation for $\pi$ based on its \emph{disjoint cycle decomposition}.  For instance, when $\pi=\langle 4\ 1\ 6\ 2\ 5\ 7\ 3\rangle$, the disjoint cycle notation is $\pi=(1,4,2)(3,6,7)(5)$ (notice the parentheses and the commas). 

\begin{figure}[htbp]
  \centering
    \begin{tikzpicture}[scale=1,>=stealth,bend angle=45]
      \tikzstyle{arc} = [draw,->]
      \tikzstyle{every loop} = []
      \node[vertex,draw] (U5) at (0,0) [label=below:$5$] {} ;
      \draw[arc] (U5) to [out=45,in=135,loop] (U5);
      \begin{scope}[xshift=3cm]
        \node[vertex,draw] (U1) at (.66,0) [label=right:$1$] {};
        \node[vertex,draw] (U4) at (-.33,.577) [label=above left:$4$] {};
        \node[vertex,draw] (U2) at (-.33,-.577) [label=below left:$2$] {};
        \draw[arc] (U1) to [out=120,in=0] (U4);
        \draw[arc] (U4) to [out=-120,in=120] (U2);
        \draw[arc] (U2) to [out=0,in=-120] (U1);
      \end{scope}
      \begin{scope}[xshift=7cm]
        \node[vertex,draw] (U3) at (.66,0) [label=right:$3$] {};
        \node[vertex,draw] (U6) at (-.33,.577) [label=above left:$6$] {};
        \node[vertex,draw] (U7) at (-.33,-.577) [label=below left:$7$] {};
        \draw[arc] (U3) to [out=120,in=0] (U6);
        \draw[arc] (U6) to [out=-120,in=120] (U7);
        \draw[arc] (U7) to [out=0,in=-120] (U3);
      \end{scope}
    \end{tikzpicture}
  \caption{The graph of the permutation $\langle 4\ 1\ 6\ 2\ 5\ 7\ 3\rangle=(1,4,2)(3,6,7)(5)$.}
  \label{fig:exemple-graph-of-a-permutation}
\end{figure}
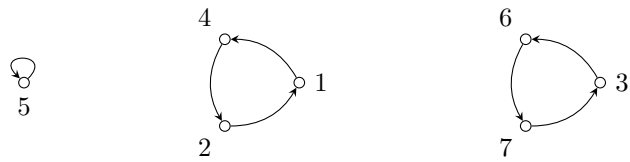

The number of cycles in $\Gamma(\pi)$ is denoted by $c(\Gamma(\pi))$, and the \emph{length} of a cycle is the number of elements it contains. A \emph{$k$-cycle} in $\Gamma(\pi)$ is a cycle of length $k$.

\begin{definition}
A \emph{signed permutation} is a permutation of $\{1, 2, \ldots, n\}$ where each element has an additional ``$+$'' or ``$-$'' sign. 
\end{definition}

The \emph{hyperoctahedral group} $S^{\pm}_n$ is the set of all 
signed permutations of $n$ elements 
with the usual function composition (the usual convention is that $\pi_{-i}=-\pi_i$). It is not mandatory for a signed permutation to have negative elements, so $S_n\subset S^{\pm}_n$ since each permutation in $S_n$ can be viewed as a signed permutation without negative elements in $S^\pm_n$. To lighten the presentation, we will conform to the tradition of omitting ``$+$'' signs when elements are positive.

\subsection{Operations on permutations}\label{sec:operations}

We now review a number of operations on permutations, which are themselves modelled as permutations.

\begin{definition}\label{def:exchange}
An \emph{exchange} $\varepsilon(i, j)$ with $1\leq i<j\leq n$ is a permutation that swaps elements in positions $i$ and $j$:
$$
\varepsilon(i, j)=\left(
\begin{array}{l}\renewcommand{\arraystretch}{1.3}
1\ \cdots\ i-1\ \fbox{$i$}\ i+1\ \cdots\ j-1\ \fbox{$j$}\ j+1\ \cdots\ n\\
\raisebox{-.05in}{$1\ \cdots\ i-1\ \fbox{$j$}\ i+1\ \cdots\ j-1\ \fbox{$i$}\ j+1\ \cdots\ n$}
\end{array}
\right).
$$
\end{definition}

\begin{definition}
A \emph{reversal} $\rho(i, j)$ with $1\leq i<j\leq n$ is a permutation that reverses the order of elements between positions $i$ and $j$:
$$
\rho(i, j)=\left(
\begin{array}{l}\renewcommand{\arraystretch}{1.3}
1\ \cdots\ i-1\ \underline{i\ \ i+1\ \cdots\ j-1\ j}\ j+1\ \cdots\ n\\
\raisebox{-.05in}{$1\ \cdots\ i-1\ j\ j-1\ \cdots\ i+1\ \ i\ j+1\ \cdots\ n$}
\end{array}
\right).
$$
\end{definition}

\begin{definition}
A \emph{signed reversal} $\overline{\rho}(i, j)$ with $1\leq i\leq j\leq n$ is a permutation that reverses both the order and the signs of elements between positions $i$ and $j$:
$$
\overline{\rho}(i, j)=\left(
\begin{array}{ccc}\renewcommand{\arraystretch}{1.3}
1\ \cdots\ i-1 & \underline{i\ \ \ \ \ i+1\ \ \ \ \ \ \cdots\ \ \ \ \ j-1\ \ \ \ \ j} & j+1\ \cdots\ n\\
1\ \cdots\ i-1 & -j\ -(j-1)\ \cdots\ -(i+1)\ -i     & j+1\ \cdots\ n
\end{array}
\right).
$$
\end{definition}

Each operation $\sigma$ transforms a permutation $\pi$ into a permutation $\pi\circ\sigma$. Setting $i=1$ in the above definitions turns those operations into \emph{prefix operations}, \ie operations whose action is restricted to the initial segment of the permutation. It can be easily seen that the effect of any operation can be mimicked by at most three prefix operations of the same kind. 
We are interested in the following two problems on permutations.

\begin{definition}
Given a permutation $\pi$ in $S^{\pm}_n$ and a set $X\subseteq S^{\pm}_n$ of allowed transformations, the problem of \emph{sorting $\pi$ by $X$} is that of finding a minimum-length sequence of elements of $X$ that transforms $\pi$ into $\iota$. The \emph{distance} of $\pi$ (with respect to $X$) is the length of such a sequence.
\end{definition}

The operations we have presented give rise to the \underline{e}xchange \underline{d}istance (denoted by $ed(\pi)$), the \underline{r}eversal \underline{d}istance (denoted by $rd(\pi)$) and the \underline{s}igned \underline{r}eversal \underline{d}istance (denoted by $srd(\pi)$), respectively, as well as the corresponding prefix variants (namely $ped(\pi)$, $prd(\pi)$ and $psrd(\pi)$). Table~\ref{tab:results-on-sorting-and-distances} summarises a selected portion of the current state of knowledge about these distances and the corresponding sorting problems.

\begin{table}[htbp]
\centering
\begin{tabular}{c|l|c|c|c}
& Operation       & Sorting      & Distance & Best approximation \\
\hline
& exchange        & \multicolumn{2}{c|}{$O(n)$~\cite{knuth-art}}                   & $1$\\
& reversal        & \multicolumn{2}{c|}{NP-hard~\cite{caprara-sorting}}            & $11/8$~\cite{berman-better}\\
& signed reversal & $O(n^{3/2})$~\cite{han-improving} & $O(n)$~\cite{bader-linear} & $1$\\
\hline
\hline
\multirow{3}{*}{\rotatebox{90}{prefix}} & exchange & \multicolumn{2}{c|}{$O(n)$~\cite{akers-star}} & $1$ \\
& reversal        & \textbf{?} & \textbf{?} & $2$~\cite{fischer-approx}\\
& signed reversal & \textbf{?} & \textbf{?} & $2$~\cite{cohen-burnt}\\
\end{tabular}
\caption{Some results on sorting permutations using various operations.}
\label{tab:results-on-sorting-and-distances}
\end{table}

\subsection{The breakpoint graph}

\citet{bafna-genome} introduced the following graph, which turned out to be an extremely useful tool to sort permutations by (possibly signed) reversals.

\begin{definition}
Given a signed permutation $\pi$ in $S^\pm_n$, transform it into an unsigned permutation $\pi'$ in $S_{2n}$ by mapping $\pi_i$ onto the sequence $(2\pi_i-1,2\pi_i)$ if $\pi_i>0$, or $(2|\pi_i|,2|\pi_i|-1)$ if $\pi_i<0$, for $1\leq i\leq n$. The \emph{breakpoint graph} of $\pi'$ is the undirected bicoloured graph $BG(\pi)$ with ordered vertex set $(\pi'_0=0,\pi'_1, \pi'_2, \ldots, \pi'_{2n},\pi'_{2n+1}=2n+1)$ and whose edge set consists of:
\begin{itemize}
\item black edges $\{\pi'_{2i}, \pi'_{2i+1}\}$ for $0\leq i\leq n$;
\item grey edges $\{\pi'_{2i}, \pi'_{2i}+1\}$ for $0\leq i\leq n$.
\end{itemize}
\end{definition}

Figure~\ref{fig:breakpoint-graph-example} shows an example of a breakpoint graph. Since each vertex in that graph has degree two, the breakpoint graph decomposes in a single way into \emph{alternating cycles}, \ie cycles that alternate black and grey edges. It can be easily seen that the breakpoint graph shown in Figure~\ref{fig:breakpoint-graph-example} decomposes into two such cycles.

\begin{figure}[htbp]
\centering
\begin{tikzpicture}[scale=.7,>=stealth]
    \foreach \p in {1, 3, ..., 17}
        \draw (\p,0) -- (\p-1,0);

    \foreach \p/\l in {1.5/-7, 3.5/3, 5.5/-1, 7.5/4, 9.5/2, 11.5/8, 13.5/-6, 15.5/-5}
        \node at (\p,-1.2) {$\l$};

    \foreach \b/\e in {0/6, 5/9, 10/3, 4/7, 8/16, 15/14, 13/2, 1/11, 12/17}
        \draw[gray] (\b,0) to  [out=90,in=90] (\e,0);

    \foreach \p/\l in {0/0, 1/14, 2/13, 3/5, 4/6, 5/2, 6/1, 7/7, 8/8, 9/3, 10/4, 11/15, 12/16, 13/12, 14/11, 15/10, 16/9, 17/17}
        \node[vertex,draw] at (\p,0) [label=below:$\l$] {};
\end{tikzpicture}
\caption{The breakpoint graph of $\langle -7\ 3\ -1\ 4\ 2\ 8\ -6\ -5\rangle$.} 
\label{fig:breakpoint-graph-example}
\end{figure}
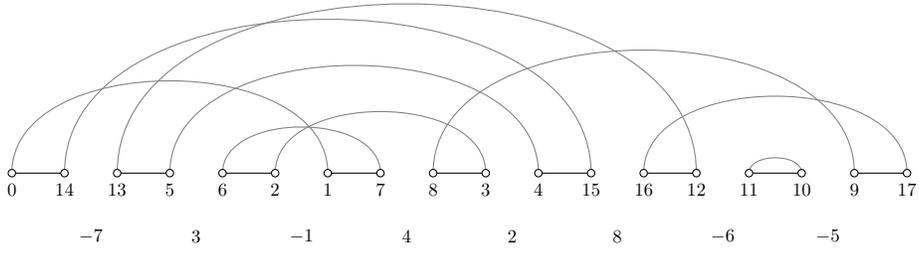

\begin{definition}\label{def:orientation-of-grey-edges}
\cite{bergeron-elementary} 
The \emph{support} of a grey edge $\{\pi'_i, \pi'_j\}$, with $i<j$, is the interval of $\pi'$ delimited by $i$ and $j$, endpoints included. A grey edge is \emph{oriented} if its support contains an odd number of elements, and \emph{nonoriented} otherwise. A cycle in $BG(\pi)$ is \emph{oriented} if it contains an oriented edge, and \emph{nonoriented} otherwise.
\end{definition}

For example, the grey edge that connects $0$ and $1$ in Figure~\ref{fig:breakpoint-graph-example} is oriented, while the one that connects $4$ and $5$ is nonoriented. The \emph{length} of a cycle in a breakpoint graph is the number of black edges it contains, and a \emph{$k$-cycle} is a cycle of length $k$. A $k$-cycle is called \emph{trivial} if $k=1$, and \emph{nontrivial} otherwise.

\begin{definition}
A signed reversal $\overline{\rho}(i,j)$ is said to \emph{act} on black edges $\{\pi'_{2i-2}$, $\pi'_{2i-1}\}$ and $\{\pi'_{2j},\pi'_{2j+1}\}$ of $BG(\pi)$. Likewise, it is said to \emph{act on one cycle} (\resp on two cycles) if both black edges on which $\overline{\rho}(i,j)$ acts belong to the same cycle (\resp to two distinct cycles) in $BG(\pi)$.
\end{definition}


\section{A new lower bound for sorting burnt pancakes}

We exploit a connection between the effect of exchanges on $\Gamma(\pi)$ and that of signed reversals on $BG(\pi)$ to derive a new lower bound on the prefix signed reversal distance of any permutation. We will need the following result by \citet{akers-star} on computing the prefix exchange distance.

\begin{theorem}\label{thm:formula-for-ped}
\cite{akers-star} For any $\pi$ in $S_n$, we have
$$ped(\pi)=
n+c(\Gamma(\pi))-2c_1(\Gamma(\pi))-\left\{
\begin{array}{ll}
0 & \mbox{if } \pi_1= 1, \\
2 & \mbox{otherwise},
\end{array}
\right.$$
where $c_1(\Gamma(\pi))$ denotes the number of $1$-cycles in $\Gamma(\pi)$.
\end{theorem}

\begin{theorem}\label{thm:lower-bound-on-psrd}
For any $\pi$ in $S^\pm_n$, we have
\begin{eqnarray}\label{eqn:lower-bound-on-psrd}
 psrd(\pi)\geq
n+1+c(BG(\pi))-2c_1(BG(\pi))-\left\{
\begin{array}{ll}
0 & \mbox{if } \pi_1= 1, \\
2 & \mbox{otherwise},
\end{array}
\right.
\end{eqnarray}
where $c_1(BG(\pi))$ denotes the number of $1$-cycles in $BG(\pi)$.
\end{theorem}
\begin{proof}
The key observation is that the action of signed reversals on the cycles of the breakpoint graph is analogous to the action of exchanges on the cycles of (the graph of) a permutation: both involve at most two distinct cycles, and can create at most one new cycle in the graph on which they act, as Figure~\ref{fig:analogy-between-exchanges-and-signed-reversals} shows. 

\begin{figure}[htbp]
\centering
\begin{tabular}{c|c|c}
\begin{tikzpicture}[>=stealth,scale=.74]
    \node[vertex,draw] at (0,.46) (pi_i) [label=left:$\pi_i$] {};
    \node[vertex,draw] at (2,.46) (pi_j) [label=right:$\pi_j$] {};
    \draw[dotted,->] (pi_i) to [out=90, in=90] (pi_j);
    \draw[dotted,<-] (pi_i) to [out=-90, in=-90] (pi_j);

    \begin{scope}[yshift=-100pt]
        \draw[->] (1,2.96) -- (1,2.16);
        \node[vertex,draw] at (0,.46) (pi_i) [label=left:$\pi_i$] {};
        \node[vertex,draw] at (2,.46) (pi_j) [label=right:$\ \pi_j$] {};
        \node[vertex] at (2,-.2) (dummy)  {};
        \draw[dotted,-] (pi_i) to [out=135, in=45,loop] (pi_i) {};
        \draw[dotted,-] (pi_j) to [out=135, in=45,loop] (pi_j) {};
    \end{scope}
\end{tikzpicture}
&
\begin{tikzpicture}[>=stealth,scale=.74]
    \foreach \p/\l in {0/$\pi'_{2i}$,1/$\pi'_{2i+1}$,4/$\pi'_{2j}$,5/$\pi'_{2j+1}$}
        \node[vertex,draw] (\p) at (\p,0) [label=below:\l] {};

    \foreach \s/\d in {0/1, 4/5}
        \draw (\s) -- (\d);

    \foreach \s/\d in {0/4, 1/5}
        \draw[dotted] (\s) to [out=90,in=90] (\d);

\begin{scope}[yshift=-100pt]
        \draw[->] (2.5,2.5) -- (2.5,1.7);
    \foreach \p/\l in {0/$\pi'_{2i}$,1/$\pi'_{2j}$,4/$\pi'_{2i+1}$,5/$\pi'_{2j+1}$}
        \node[vertex,draw] (\p) at (\p,0) [label=below:\l] {};

    \foreach \s/\d in {0/1, 4/5}
        \draw (\s) -- (\d);

    \foreach \s/\d in {0/1, 4/5}
        \draw[dotted] (\s) to [out=90,in=90] (\d);

\end{scope}
\end{tikzpicture}
&
\begin{tikzpicture}[>=stealth,scale=.74]
    \foreach \p/\l in {0/$\pi'_{2i}$,1/$\pi'_{2i+1}$,4/$\pi'_{2j}$,5/$\pi'_{2j+1}$}
        \node[vertex,draw] (\p) at (\p,0) [label=below:\l] {};

    \foreach \s/\d in {0/1, 4/5}
        \draw (\s) -- (\d);

    \foreach \s/\d in {0/5, 1/4}
        \draw[dotted] (\s) to [out=90,in=90] (\d);

\begin{scope}[yshift=-100pt]
        \draw[->] (2.5,2.5) -- (2.5,1.7);
    \foreach \p/\l in {0/$\pi'_{2i}$,1/$\pi'_{2j}$,4/$\pi'_{2i+1}$,5/$\pi'_{2j+1}$}
        \node[vertex,draw] (\p) at (\p,0) [label=below:\l] {};

    \foreach \s/\d in {0/1, 4/5}
        \draw (\s) -- (\d);

    \foreach \s/\d in {0/5, 1/4}
        \draw[dotted] (\s) to [out=90,in=90] (\d);

\end{scope}
\end{tikzpicture}
 \\
$(a)$ & $(b)$ & $(c)$ 
\end{tabular}

\caption{$(a)$ The effect of an exchange on a cycle of $\Gamma(\pi)$, $(b)$ the effect of a signed reversal on an oriented cycle and $(c)$ on a nonoriented cycle of $BG(\pi)$. Dotted edges in those drawings stand for (alternating) paths.}
\label{fig:analogy-between-exchanges-and-signed-reversals}
\end{figure}
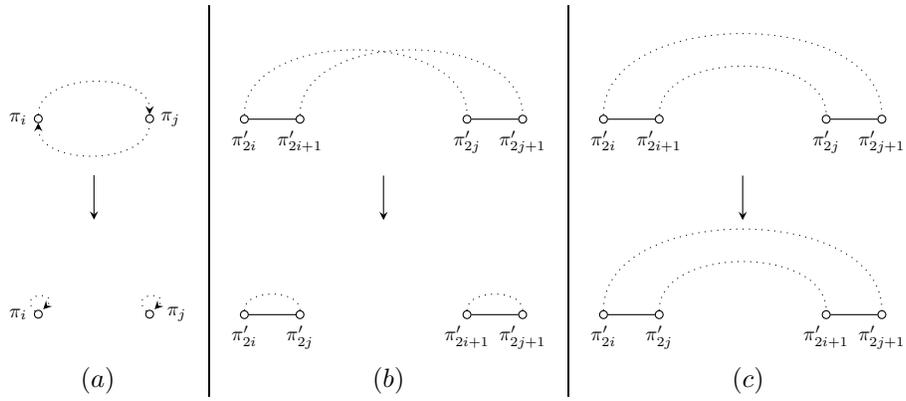

The analogy obviously still holds under the prefix restriction, and the proof then follows from Theorem~\ref{thm:formula-for-ped}.
\end{proof}

Note that, as observed by \citet{hannenhalli-transforming} and as shown in Figure~\ref{fig:analogy-between-exchanges-and-signed-reversals}~$(c)$, it is not always possible to split a cycle in $BG(\pi)$ using a signed reversal, whereas it is always possible to split a cycle in $\Gamma(\pi)$ using an exchange (hence the lower bound instead of an equality).


\section{Sorting simple permutations in polynomial time}

We now turn our attention to an important class of signed permutations, which proved crucial in solving the signed version of sorting by unrestricted reversals in polynomial time (see \citet{hannenhalli-transforming}), and show how to sort those permutations by prefix signed reversals in polynomial time.

\begin{definition}
A signed permutation $\pi$ is \emph{simple} if $BG(\pi)$ contains only cycles of length at most $2$.
\end{definition}

Our analysis is based exclusively on simple permutations; therefore, we need to ensure that the sequences of prefix signed reversals we use will transform any given simple permutation into another simple permutation.

\begin{definition}
A sequence of signed reversals applied to a simple permutation $\pi$ is \emph{conservative} if it transforms $\pi$ into a simple permutation $\sigma$.
\end{definition}

We wish to stress that we only require $\sigma$ to be simple: we allow intermediate permutations obtained in the process of transforming $\pi$ into $\sigma$ not to be simple.

\begin{definition}
Let $g(\pi)$ denote the right-hand side of lower bound~\eqref{eqn:lower-bound-on-psrd}; an \emph{$(x,y)$-sequence} is a sequence of $x$ prefix signed reversals transforming a permutation $\pi$ into a permutation $\sigma$ with $g(\pi)-g(\sigma)=y$. It is \emph{optimal} if $x=y$.
\end{definition}

\subsection{Components of the breakpoint graph}

We will need the following definitions and results of \citet{hannenhalli-transforming}.

\begin{definition}
Two distinct grey edges in the breakpoint graph \emph{interleave} if their supports overlap but do not contain each other. Likewise, two distinct cycles in the breakpoint graph \emph{interleave} if either cycle contains a grey edge that interleaves with a grey edge of the other cycle.
\end{definition}

\begin{definition}\label{def:components-and-orientation-in-BG}
Let $H(\pi)$ be the graph whose vertices are the cycles in $BG(\pi)$ and whose edges connect two vertices if the corresponding cycles interleave. A \emph{component} of $BG(\pi)$ is a connected component of $H(\pi)$; it is \emph{oriented} if a vertex of that component in $H(\pi)$ corresponds to an oriented cycle in $BG(\pi)$, and \emph{nonoriented} otherwise.
\end{definition}

\begin{lemma}\label{lemma:reversals-change-orientation-of-interleaving-cycles}
\cite{hannenhalli-transforming} 
A signed reversal acting on a given cycle $C$ in $BG(\pi)$ changes the orientation of every cycle in $BG(\pi)$ that interleaves with $C$ (\ie it transform every nonoriented (\resp nonoriented) cycle that interleaves with $C$ into a nonoriented (\resp oriented) cycle).
\end{lemma}

\begin{lemma}\label{lemma:every-edge-is-interleaved}
\cite{hannenhalli-transforming} 
Every grey edge of a nontrivial cycle in $BG(\pi)$ interleaves with another grey edge.
\end{lemma}

Lemma~\ref{lemma:every-edge-is-interleaved} implies in particular that if $\pi$ is a simple permutation, then every nontrivial nonoriented cycle in $BG(\pi)$ interleaves with another nontrivial cycle. In the following, we sometimes abuse language by saying that we sort a cycle or a component, which actually means that we transform a $k$-cycle or a component involving $k$ black edges in $BG(\pi)$ into a collection of $k$ $1$-cycles.

\subsection{Preliminary results}

In the following, we will refer to the cycle in $BG(\pi)$ that contains black edge $\{\pi'_0, \pi'_1\}$ as the \emph{leftmost cycle}, and to the component that contains the leftmost cycle as the \emph{leftmost component}.

\begin{definition}
A signed reversal is \emph{proper} if it increases the number of cycles in $BG(\pi)$ by one.
\end{definition}

The following observation will be crucial.

\begin{lemma}\label{lemma:mimicking-proper-reversals-is-optimal-for-simple-permutations}
For any simple permutation $\pi$, any minimal sequence of prefix reversals that mimics a proper reversal is both conservative and optimal.
\end{lemma}
\begin{proof}
A proper reversal $\overline{\rho}(i,j)$ on $\pi$ splits a $2$-cycle into two $1$-cycles; if $i=1$, then the resulting permutation $\sigma$ now fixes $1$, and lower bound~\eqref{eqn:lower-bound-on-psrd} has decreased by $1$. Otherwise, the status of the first element in $\pi$ and $\sigma$ is the same, and there is a sequence of three prefix reversals which mimics the effect of $\overline{\rho}(i,j)$
and decreases lower bound~\eqref{eqn:lower-bound-on-psrd} by $3$. Therefore, the resulting sequence of prefix reversals is optimal, and clearly conservative.
\end{proof}

As a result, if $BG(\pi)$ contains an oriented component, then we can sort that component optimally in polynomial time (see \citet{tannier-advances} for more details). Therefore, the only remaining cases we need to examine are the cases where $\pi$ admits no proper reversal, or equivalently where $BG(\pi)$ contains no oriented cycle, distinguishing between the case where $\pi_1\neq 1$ and $\pi_1=1$.

\begin{lemma}\label{lemma:simple-permutation-leftmost-component-is-a-hurdle}
Let $\pi$ be a simple permutation; if $\pi_1\neq 1$ and $BG(\pi)$ contains no oriented cycle, then $\pi$ admits a conservative $(1,0)$-sequence.
\end{lemma}
\begin{proof}
A prefix signed reversal acting on the leftmost cycle will flip the orientation of any cycle it interleaves (Lemma~\ref{lemma:reversals-change-orientation-of-interleaving-cycles}), thereby guaranteeing the creation of at least one oriented cycle in $BG(\pi)$ and in particular transforming the leftmost component into an oriented component. Note that lower bound~\eqref{eqn:lower-bound-on-psrd} is unaffected by that move.
\end{proof}

\begin{lemma}\label{lemma:simple-permutation-fixing-one}
Let $\pi$ be a simple permutation; if $\pi_1=1$ and $BG(\pi)$ contains no oriented cycle, then $\pi$ admits a conservative $(2,2)$-sequence.
\end{lemma}
\begin{proof}
If $\pi_1=1$ and $BG(\pi)$ contains no oriented $2$-cycle, then any nonoriented component $BG(\pi)$ contains can be transformed into an oriented leftmost component as follows. Pick the leftmost cycle $C_1$ in any nonoriented component; by Lemma~\ref{lemma:every-edge-is-interleaved}, $C_1$ interleaves with another nonoriented cycle, say $C_2$. If $C_2$ contains black edges $\{\pi'_{2j-2},\pi'_{2j-1}\}$ and $\{\pi'_{2\ell-2},\pi'_{2\ell-1}\}$, then applying $\overline{\rho}(1,j-1)$ followed by $\overline{\rho}(1,\ell-1)$ transforms $\pi$ into another simple permutation with an oriented leftmost component, as Figure~\ref{fig:orienting-nonoriented-components} shows.

\begin{figure}[htbp]
 \centering
\begin{tikzpicture}[scale=.57,>=stealth]
    \foreach \p in {1,3,5,7,9}
        \draw (\p,0) -- (\p-1,0);

    \foreach \p/\n/\l in {0/0/$\pi'_0$, 1/a/$\pi'_1$, 2/b/$\pi'_{2i-2}$, 3/c/$\pi'_{2i-1}$, 4/d/$\pi'_{2j-2}$, 5/e/$\pi'_{2j-1}$, 6/f/$\pi'_{2k-2}$, 7/g/$\pi'_{2k-1}$, 8/h/$\pi'_{2\ell-2}$, 9/i/$\pi'_{2\ell-1}$}
        \node[vertex] at (\p,0) [draw,circle] (\n) [label=below:\l] {};

    \foreach \s/\d in {0/a, b/g, c/f, d/i, e/h}
        \draw[gray] (\s) to  [out=90,in=90] (\d);

    \draw (.5,-.7) -- (4.5,-.7);
\begin{scope}[xshift=320pt]
    \foreach \p in {1,3,5,7,9}
        \draw (\p,0) -- (\p-1,0);
    \foreach \s/\d in {0/4, 1/9, 2/6, 3/7, 5/8}
        \draw[gray] (\s,0) to  [out=90,in=90] (\d,0);
    \foreach \p/\n/\l in {0/0/$\pi'_0$, 1/d/$\pi'_{2j-2}$, 2/c/$\pi'_{2i-1}$, 3/b/$\pi'_{2i-2}$, 4/a/$\pi'_{1}$, 5/e/$\pi'_{2j-1}$, 6/f/$\pi'_{2k-2}$, 7/g/$\pi'_{2k-1}$, 8/h/$\pi'_{2\ell-2}$, 9/i/$\pi'_{2\ell-1}$}
        \node[vertex] at (\p,0) [draw,circle] (\n) [label=below:\l] {};

    \draw (.5,-.7) -- (8.5,-.7);
\draw[->] (-1.5,0) -- (-0.5,0);
\end{scope}
\begin{scope}[xshift=320pt,yshift=-100pt]
    \foreach \p in {1,3,5,7,9}
        \draw (\p,0) -- (\p-1,0);
    \foreach \s/\d in {0/5, 1/4, 8/9, 2/6, 3/7}
        \draw[gray] (\s,0) to  [out=90,in=90] (\d,0);
    \foreach \p/\n/\l in {0/0/$\pi'_0$, 1/d/$\pi'_{2\ell-2}$, 2/c/$\pi'_{2k-1}$, 3/b/$\pi'_{2k-2}$, 4/a/$\pi'_{2j-1}$, 5/e/$\pi'_{1}$, 6/f/$\pi'_{2i-2}$, 7/g/$\pi'_{2i-1}$, 8/h/$\pi'_{2j-2}$, 9/i/$\pi'_{2\ell-1}$}
        \node[vertex] at (\p,0) [draw,circle] (\n) [label=below:\l] {};
\draw[->] (-1.5,0) -- (-0.5,0);
\end{scope}
\end{tikzpicture}
\caption{An optimal conservative sequence of prefix signed reversals that transforms a nonoriented component into an oriented leftmost component.}
\label{fig:orienting-nonoriented-components}
\end{figure}
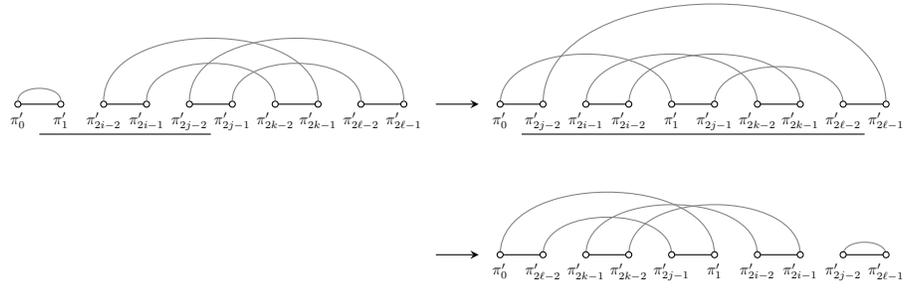

Neither the number of cycles nor the number of $1$-cycles in the breakpoint graph is affected, but the resulting permutation no longer fixes $1$, so lower bound~\eqref{eqn:lower-bound-on-psrd} decreases by $2$.	
\end{proof}

\subsection{Computing $psrd$ for simple permutations}

We have proved the existence of optimal conservative sequences for every simple permutation, except in the case where $\pi_1\neq 1$ and the leftmost component is nonoriented. In this section, we prove that the strategy proposed in Lemma~\ref{lemma:simple-permutation-leftmost-component-is-a-hurdle} is optimal (Lemma~\ref{lemma:josef-lemma-three}), and derive a formula for computing the prefix signed reversal distance of simple permutations (Theorem~\ref{thm:formula-for-psrd-of-simple-permutations}).

We will refer to prefix signed reversals that act on two nontrivial cycles of the breakpoint graph as \emph{merging moves}, and to those that split the leftmost cycle into two cycles, at least one of which is trivial, as \emph{splitting moves}. It can be easily seen from Theorem~\ref{thm:lower-bound-on-psrd} that any sequence of prefix signed reversals that is to outperform the strategy proposed in Lemma~\ref{lemma:simple-permutation-leftmost-component-is-a-hurdle} must consist solely of these types of moves, and eventually lead to an oriented leftmost $2$-cycle, a step that must precede the creation of two new trivial cycles. Therefore, our proof will consist in showing that this strategy will fail to orient the leftmost component.

We have already defined orientation for grey edges (Definition~\ref{def:orientation-of-grey-edges} page~\pageref{def:orientation-of-grey-edges}). We will also need an analogous definition for black edges by \citet{tannier-subquadratic}.

\begin{definition}\label{def:orientation-of-black-edges}
\cite{tannier-subquadratic} A black edge $\{\pi'_{2i}, \pi'_{2i+1}\}$ in $BG(\pi)$ is \emph{oriented} if $\pi_i$ and $\pi_{i+1}$ have opposite signs, and \emph{nonoriented} otherwise.
\end{definition}

\begin{definition}\cite{hannenhalli-transforming}
The smallest interval that contains the leftmost and rightmost elements of a given component $\mathscr C$ in $BG(\pi)$ is denoted by
$$Extent(\mathscr C)=\left[\min_{C\in\mathscr C}\min_{\pi'_i\in C}i,\max_{C\in\mathscr C}\max_{\pi'_j\in C}j \right].$$
\end{definition}

Components can be ordered by inclusion based on their extent. We will be interested mainly in minimal components, defined below.

\begin{definition}\cite{hannenhalli-transforming}\label{def:minimal-component}
 A component $\mathscr C$ of $BG(\pi)$ is \emph{minimal} if every cycle $C$ with $support(C)\subseteq Extent(\mathscr C)$ belongs to $\mathscr C$.
\end{definition}

Graphically speaking, a minimal component is an ``innermost'' one. \citet{tannier-subquadratic} note that $BG(\pi^{-1})$ can be obtained from $BG(\pi)$ by exchanging positions and elements in $\pi'$ as well as edge colours (\ie black (\resp grey) edges in $BG(\pi)$ become grey (\resp black) edges in $BG(\pi^{-1})$). As a consequence, cycles in $BG(\pi)$ and in $BG(\pi^{-1})$ are in one-to-one correspondence, and we denote $C^{-1}$ the cycle in $BG(\pi^{-1})$ onto which a given cycle $C$ in $BG(\pi)$ is mapped. We show below that this mapping extends to whole components of the breakpoint graph.

\begin{lemma}\label{lemma:josef-lemma-one}
A cycle $C$ belongs to a component $\mathscr C$ in $BG(\pi)$ if and only if $C^{-1}$ belongs to $\mathscr C^{-1}$ in $BG(\pi^{-1})$.
\end{lemma}
\begin{proof}
We examine components by their inclusion order, starting with minimal ones and removing them as we go to proceed by induction. 

We refer to the pairs $\{\pi'_{2i-1},\pi'_{2i}\}$ for $1 \leq i \leq n$ as \emph{white edges}, which form the same set in both $BG(\pi)$ and $BG(\pi^{-1})$. The alternating path made of white and grey (\resp black) edges in $BG(\pi)$ will be referred to as the \emph{WG-path} (\resp \emph{WB-path}), which when starting with the leftmost vertex of $BG(\pi)$ visits the vertices of $BG(\pi)$ in the natural order (\resp in the order in which they appear in $\pi'$).

We now show that a minimal component $\mathscr C$ in $BG(\pi)$ corresponds to a minimal component $\mathscr C^{-1}$ in $BG(\pi^{-1})$. If $Extent(\mathscr C)=[i, j]$, then vertices $\{\pi'_{2i}, \ldots, \pi'_{2j-1}\}$ induce a sub-path of the WG-path in $BG(\pi)$, which is mapped onto a WB-path in $BG(\pi^{-1})$ and implies that cycles of $\mathscr C^{-1}$ do not interleave with cycles outside $\mathscr C^{-1}$. To see that the cycles we obtain in that way belong to the same connected component in $BG(\pi^{-1})$, assume on the contrary that some cycles were mapped onto an additional minimal component, say $\mathscr C'$, in $BG(\pi^{-1})$. By the above argument, cycles in $\mathscr C'^{-1}$ do not interleave with other cycles in $BG((\pi^{-1})^{-1}) = BG(\pi)$, a contradiction.
\end{proof}

The following result will also be useful.

\begin{lemma}\label{lemma:josef-lemma-two}
A component $\mathscr C$ in $BG(\pi)$ is oriented if and only if it contains an oriented black edge.
\end{lemma}
\begin{proof}
As in the proof of Lemma~\ref{lemma:josef-lemma-one}, we examine components by their inclusion order, starting with minimal ones and removing them as we go to proceed by induction. 
Recall (Definition~\ref{def:components-and-orientation-in-BG} page~\pageref{def:components-and-orientation-in-BG}) that $\mathscr C$ is oriented if it contains an oriented grey edge, which corresponds to a pair of elements of opposite signs in $\pi$. If $\mathscr C$ is minimal, then it contains two elements of opposite signs in $\pi$ if and only if it contains a pair of \emph{adjacent} elements of opposite signs, which themselves correspond to an oriented black edge (Definition~\ref{def:orientation-of-black-edges}). 

As observed by~\citet{bergeron-common}, the elements of $\pi$ that ``frame'' $\mathscr C$ (\ie $\pi_{i/2}$ and $\pi_{(j+1)/2}$, if $Extent(\mathscr C)=[i,j]$) have the same sign. We can then remove $\mathscr C$ from $BG(\pi)$, renumber the elements appropriately, and handle the remaining components in the same way.
\end{proof}

Since oriented black (\resp grey) edges in $BG(\pi)$ become oriented grey (\resp black) edges in $BG(\pi^{-1})$, Lemma~\ref{lemma:josef-lemma-two} implies that the orientation of components in $BG(\pi)$ is also preserved in $BG(\pi^{-1})$.

\begin{lemma}\label{lemma:josef-lemma-three}
Let $\pi$ be simple permutation which does not fix $1$ and whose leftmost component is nonoriented. If a sequence of merging and splitting moves transforms $\pi$ into a permutation $\sigma$ whose leftmost cycle is a $2$-cycle, then $\sigma$ is simple and the leftmost component of $BG(\sigma)$ is nonoriented.
\end{lemma}
\begin{proof}
Splitting moves extract $1$-cycles from the leftmost cycle, while merging moves merge it with $2$-cycles, so clearly $\sigma$ is simple. To prove that its leftmost component is nonoriented, we first show that the sets of black edges in $BG(\pi)$ and in $BG(\sigma)$ differ only in the black edges that belonged to $2$-cycles in $BG(\pi)$ and became $1$-cycles in $BG(\sigma)$. 

Grey edges in $BG(\sigma^{-1})$ connect pairs of elements that appear at the same positions as in $BG(\pi^{-1})$, except for edges that originally belonged to $2$-cycles that were turned into $1$-cycles. Each $2$-cycle in $BG(\sigma^{-1})$ corresponds to a $2$-cycle in $BG(\pi^{-1})$ on the same quadruple of positions (since by the transformation described right after Definition~\ref{def:minimal-component}, black and grey edges are exchanged in the process of transforming $BG(\pi)$ into $BG(\pi^{-1})$ and conversely). The orientation of any $2$-cycle in $BG(\sigma^{-1})$ is the same as the orientation of the corresponding $2$-cycle in $BG(\pi^{-1})$, and there is no new pair of interleaving $2$-cycles in $BG(\sigma^{-1})$. This together with Lemma~\ref{lemma:josef-lemma-two} implies that the leftmost component of $BG(\sigma)$ is nonoriented.
\end{proof}

We now have everything we need to prove a formula for computing the prefix signed reversal distance of simple permutations.

\begin{theorem}\label{thm:formula-for-psrd-of-simple-permutations}
For every simple permutation $\pi$ in $S^{\pm}_n$, we have:
$$psrd(\pi)
=
n+1+c(BG(\pi))-2c_1(BG(\pi))+t(\pi)-\left\{
\begin{array}{ll}
0 & \mbox{if } \pi_1= 1 \\
2 & \mbox{otherwise}
\end{array}
\right.,
$$
where $t(\pi)=1$ if $\pi_1\neq 1$ and the leftmost component of $BG(\pi)$ is nonoriented, and $0$ otherwise.
\end{theorem}
\begin{proof}
The upper bound follows from the fact that there exists an optimal conservative sequence for dealing with every simple permutation, except when $\pi_1\neq 1$ and $BG(\pi)$ contains no oriented cycle; however, a single prefix signed reversal turns the leftmost component into an oriented component (Lemma~\ref{lemma:simple-permutation-leftmost-component-is-a-hurdle}). 
This situation cannot occur more than once in the sorting process, since once the leftmost component has been sorted, either the resulting permutation is $\iota$ or we can sort every remaining oriented component of $BG(\pi)$ optimally -- or create a leftmost oriented component in $BG(\pi)$ if no oriented component exists (Lemma~\ref{lemma:simple-permutation-fixing-one}).

Finally, if $\pi_1\neq 1$ and the leftmost component of $BG(\pi)$ is nonoriented, then no sorting sequence can outperform the strategy proposed in Lemma~\ref{lemma:simple-permutation-leftmost-component-is-a-hurdle} (see Lemma~\ref{lemma:josef-lemma-three}), which implies the desired lower bound and completes the proof.
\end{proof}

\subsection{The sorting algorithm}

Algorithm~\ref{algorithm:SimpleBurntPancakeFlipping} outlines how to sort simple permutations by prefix signed reversals in polynomial time. \citet{tannier-advances} cover step $3$ in details (they sort oriented components using arbitrary signed reversals, but as we have seen these can be mimicked by optimal sequences of prefix signed reversals (Lemma~\ref{lemma:mimicking-proper-reversals-is-optimal-for-simple-permutations})), while step $5$ can be achieved either by applying a reversal on the leftmost cycle, if $\mathcal D$ is the leftmost component (Lemma~\ref{lemma:simple-permutation-leftmost-component-is-a-hurdle}), or by applying the $2$-move sequence proposed in Lemma~\ref{lemma:simple-permutation-fixing-one}. The algorithm can be implemented so as to run in $O(n^{3/2})$ time (see \citet{tannier-advances,han-improving}), while the distance can be computed in $O(n)$ time (see \citet{bader-linear}).

\begin{algorithm}
\caption{{\sc SimpleBurntPancakeFlipping($\pi$)}}
\label{algorithm:SimpleBurntPancakeFlipping}
\algorithmicinputdata a simple permutation $\pi$

\algorithmicresult the identity permutation
\begin{algorithmic}[1]
\WHILE {$\pi\neq\iota$}
    \IF {$BG(\pi)$ contains an oriented component $\mathcal C$}
        \STATE sort $\mathcal C$;
    \ELSE
        \STATE orient any nonoriented component $\mathcal D$;
    \ENDIF
\ENDWHILE
\end{algorithmic}
\end{algorithm}

\section{Conclusions}

We proved a new lower bound on the minimum number of prefix signed reversals needed to sort any signed permutation of $n$ elements, whereas the exact computation of that number remains an open problem. Using this lower bound, we were able to show that an important class of permutations, known as ``simple permutations'', could be sorted in polynomial time, and proposed both a sorting algorithm and a formula for computing the minimum number of required operations.

\citet{hannenhalli-transforming} proved that every permutation $\pi$ could be transformed into a simple permutation $\tilde{\pi}$ in such a way that $srd(\pi)=srd(\tilde{\pi})$ (see \citet{gog-fast} for a $O(n)$ time algorithm for transforming $\pi$ into $\tilde{\pi}$, and a $O(n\log n)$ time algorithm for recovering the original permutation). Unfortunately, the transformation does not preserve the prefix signed reversal distance, as shown by the following counter-example: if $\pi=\langle 2\ 1\rangle$, then the corresponding simple permutation is $\tilde{\pi}=\langle 3\ 2\ 1\rangle$, but it can be verified that $psrd(\pi)=3$ and $psrd(\tilde{\pi})=5$. Therefore, Algorithm~\ref{algorithm:SimpleBurntPancakeFlipping} cannot immediately be used to sort an arbitrary permutation, but since every sequence of signed reversals on a simple permutation can be used to sort the original permutation~\cite{hannenhalli-transforming}, we have $psrd(\pi)\leq psrd(\tilde{\pi})$. 
Moreover, we believe that our contributions should be useful for designing improved approximation or exact algorithms for solving the burnt pancake flipping problem, as well as for getting insight into its computational complexity. The unsigned version of sorting by prefix reversals (\ie the original pancake flipping problem) may also benefit from our results, since both variants are strongly connected (see \citet{hannenhalli-tocut} for more details).

\bibliographystyle{elsarticle-num-names}
\bibliography{burnt-pancakes-polynomial}

\def\cprime{$'$}
\begin{thebibliography}{25}
\providecommand{\natexlab}[1]{#1}
\providecommand{\url}[1]{\texttt{#1}}
\providecommand{\urlprefix}{URL }
\expandafter\ifx\csname urlstyle\endcsname\relax
  \providecommand{\doi}[1]{doi:\discretionary{}{}{}#1}\else
  \providecommand{\doi}[1]{doi:\discretionary{}{}{}\begingroup
  \urlstyle{rm}\url{#1}\endgroup}\fi
\providecommand{\bibinfo}[2]{#2}

\bibitem[{Dweighter(1975)}]{dweighter-elementary}
\bibinfo{author}{H.~Dweighter}, \bibinfo{title}{Elementary problems and
  solutions}, \bibinfo{journal}{American Mathematical Monthly}
  \bibinfo{volume}{82}~(\bibinfo{number}{10}) (\bibinfo{year}{1975})
  \bibinfo{pages}{296}, \bibinfo{note}{problem E2569}.

\bibitem[{Gates and Papadimitriou(1979)}]{gates-bounds}
\bibinfo{author}{W.~H. Gates}, \bibinfo{author}{C.~H. Papadimitriou},
  \bibinfo{title}{Bounds for sorting by prefix reversal},
  \bibinfo{journal}{Discrete Mathematics}
  \bibinfo{volume}{27}~(\bibinfo{number}{1}) (\bibinfo{year}{1979})
  \bibinfo{pages}{47--57}, ISSN \bibinfo{issn}{0012-365X}.

\bibitem[{Gy\"ori and Tur\'an(1978)}]{gyori-stack}
\bibinfo{author}{E.~Gy\"ori}, \bibinfo{author}{G.~Tur\'an},
  \bibinfo{title}{Stack of Pancakes}, \bibinfo{journal}{Studia Scientiarum
  Mathematicarum Hungarica} \bibinfo{volume}{13} (\bibinfo{year}{1978})
  \bibinfo{pages}{133--137}.

\bibitem[{Fischer and Ginzinger(2005)}]{fischer-approx}
\bibinfo{author}{J.~Fischer}, \bibinfo{author}{S.~W. Ginzinger},
  \bibinfo{title}{A $2$-Approximation Algorithm for Sorting by Prefix
  Reversals}, in: \bibinfo{editor}{G.~S. Brodal}, \bibinfo{editor}{S.~Leonardi}
  (Eds.), \bibinfo{booktitle}{Proceedings of the Thirteenth Annual European
  Symposium on Algorithms (ESA)}, vol. \bibinfo{volume}{3669} of
  \emph{\bibinfo{series}{Lecture Notes in Computer Science}},
  \bibinfo{publisher}{Springer-Verlag}, \bibinfo{pages}{415--425},
  \bibinfo{year}{2005}.

\bibitem[{Cohen and Blum(1995)}]{cohen-burnt}
\bibinfo{author}{D.~S. Cohen}, \bibinfo{author}{M.~Blum}, \bibinfo{title}{On
  the Problem of Sorting Burnt Pancakes}, \bibinfo{journal}{Discrete Applied
  Mathematics} \bibinfo{volume}{61}~(\bibinfo{number}{2})
  (\bibinfo{year}{1995}) \bibinfo{pages}{105--120}.

\bibitem[{Haynes et~al.(2008)Haynes, Broderick, Brown, Butner, Dickson, Harden,
  Heard, Jessen, Malloy, Ogden, Rosemond, Simpson, Zwack, Campbell, Eckdahl,
  Heyer, and Poet}]{haynes-engineering}
\bibinfo{author}{K.~Haynes}, \bibinfo{author}{M.~Broderick},
  \bibinfo{author}{A.~Brown}, \bibinfo{author}{T.~Butner},
  \bibinfo{author}{J.~Dickson}, \bibinfo{author}{W.~L. Harden},
  \bibinfo{author}{L.~Heard}, \bibinfo{author}{E.~Jessen},
  \bibinfo{author}{K.~Malloy}, \bibinfo{author}{B.~Ogden},
  \bibinfo{author}{S.~Rosemond}, \bibinfo{author}{S.~Simpson},
  \bibinfo{author}{E.~Zwack}, \bibinfo{author}{A.~M. Campbell},
  \bibinfo{author}{T.~Eckdahl}, \bibinfo{author}{L.~Heyer},
  \bibinfo{author}{J.~Poet}, \bibinfo{title}{Engineering bacteria to solve the
  Burnt Pancake Problem}, \bibinfo{journal}{Journal of Biological Engineering}
  \bibinfo{volume}{2}~(\bibinfo{number}{1}) (\bibinfo{year}{2008})
  \bibinfo{pages}{8}, ISSN \bibinfo{issn}{1754-1611},
  \doi{\bibinfo{doi}{10.1186/1754-1611-2-8}},
  \urlprefix\url{http://www.jbioleng.org/content/2/1/8}.

\bibitem[{Lakshmivarahan et~al.(1993)Lakshmivarahan, Jwo, and
  Dhall}]{lakshmivarahan-symmetry}
\bibinfo{author}{S.~Lakshmivarahan}, \bibinfo{author}{J.-S. Jwo},
  \bibinfo{author}{S.~K. Dhall}, \bibinfo{title}{Symmetry in interconnection
  networks based on {C}ayley graphs of permutation groups: A survey},
  \bibinfo{journal}{Parallel Computing}
  \bibinfo{volume}{19}~(\bibinfo{number}{4}) (\bibinfo{year}{1993})
  \bibinfo{pages}{361--407}.

\bibitem[{Caprara(1999)}]{caprara-sorting}
\bibinfo{author}{A.~Caprara}, \bibinfo{title}{Sorting permutations by reversals
  and Eulerian cycle decompositions}, \bibinfo{journal}{{SIAM} Journal on
  Discrete Mathematics} \bibinfo{volume}{12}~(\bibinfo{number}{1})
  (\bibinfo{year}{1999}) \bibinfo{pages}{91--110 (electronic)}, ISSN
  \bibinfo{issn}{1095-7146}.

\bibitem[{Hannenhalli and Pevzner(1999)}]{hannenhalli-transforming}
\bibinfo{author}{S.~Hannenhalli}, \bibinfo{author}{P.~A. Pevzner},
  \bibinfo{title}{Transforming Cabbage into Turnip: Polynomial Algorithm for
  Sorting Signed Permutations by Reversals}, \bibinfo{journal}{Journal of the
  ACM} \bibinfo{volume}{46}~(\bibinfo{number}{1}) (\bibinfo{year}{1999})
  \bibinfo{pages}{1--27}, ISSN \bibinfo{issn}{0004-5411}.

\bibitem[{Fertin et~al.(2009)Fertin, Labarre, Rusu, Tannier, and
  Vialette}]{fertin-combinatorics}
\bibinfo{author}{G.~Fertin}, \bibinfo{author}{A.~Labarre},
  \bibinfo{author}{I.~Rusu}, \bibinfo{author}{E.~Tannier},
  \bibinfo{author}{S.~Vialette}, \bibinfo{title}{Combinatorics of Genome
  Rearrangements}, Computational Molecular Biology, \bibinfo{publisher}{The MIT
  Press}, \bibinfo{year}{2009}.

\bibitem[{Labarre(2008)}]{labarre-edit}
\bibinfo{author}{A.~Labarre}, \bibinfo{title}{Edit Distances and Factorisations
  of Even Permutations}, in: \bibinfo{editor}{D.~Halperin},
  \bibinfo{editor}{K.~Mehlhorn} (Eds.), \bibinfo{booktitle}{Proceedings of the
  Sixteenth Annual European Symposium on Algorithms (ESA)}, vol.
  \bibinfo{volume}{5193} of \emph{\bibinfo{series}{Lecture Notes in Computer
  Science}}, \bibinfo{publisher}{Springer-Verlag}, ISBN
  \bibinfo{isbn}{978-3-540-87743-1}, \bibinfo{pages}{635--646},
  \bibinfo{year}{2008}.

\bibitem[{Bj\"{o}rner and Brenti(2005)}]{bjorner-combinatorics}
\bibinfo{author}{A.~Bj\"{o}rner}, \bibinfo{author}{F.~Brenti},
  \bibinfo{title}{Combinatorics of {C}oxeter Groups}, vol.
  \bibinfo{volume}{231} of \emph{\bibinfo{series}{Graduate Texts in
  Mathematics}}, chap. \bibinfo{chapter}{8: Combinatorial Descriptions},
  \bibinfo{publisher}{Springer-Verlag}, \bibinfo{year}{2005}.

\bibitem[{Wielandt(1964)}]{wielandt-finite}
\bibinfo{author}{H.~Wielandt}, \bibinfo{title}{Finite permutation groups},
  Translated from German by R. Bercov, \bibinfo{publisher}{Academic Press},
  \bibinfo{address}{New York}, \bibinfo{year}{1964}.

\bibitem[{Bafna and Pevzner(1993)}]{bafna-genome}
\bibinfo{author}{V.~Bafna}, \bibinfo{author}{P.~A. Pevzner},
  \bibinfo{title}{Genome rearrangements and sorting by reversals}, in:
  \bibinfo{booktitle}{Proceedings of the Thirty-Fourth Annual Symposium on
  Foundations of Computer Science (FOCS)}, \bibinfo{publisher}{ACM/SIAM},
  \bibinfo{address}{Palo Alto, Los Alamitos, CA}, ISBN
  \bibinfo{isbn}{0-89871-349-8}, \bibinfo{pages}{148--157},
  \bibinfo{year}{1993}.

\bibitem[{Bergeron(2005)}]{bergeron-elementary}
\bibinfo{author}{A.~Bergeron}, \bibinfo{title}{A very elementary presentation
  of the {H}annenhalli-{P}evzner theory}, \bibinfo{journal}{Discrete Applied
  Mathematics} \bibinfo{volume}{146}~(\bibinfo{number}{2})
  (\bibinfo{year}{2005}) \bibinfo{pages}{134--145}, ISSN
  \bibinfo{issn}{0166-218X}.

\bibitem[{Akers et~al.(1987)Akers, Krishnamurthy, and Harel}]{akers-star}
\bibinfo{author}{S.~B. Akers}, \bibinfo{author}{B.~Krishnamurthy},
  \bibinfo{author}{D.~Harel}, \bibinfo{title}{The Star Graph: An Attractive
  Alternative to the $n$-Cube}, in: \bibinfo{booktitle}{Proceedings of the
  Fourth International Conference on Parallel Processing (ICPP)},
  \bibinfo{publisher}{Pennsylvania State University Press},
  \bibinfo{pages}{393--400}, \bibinfo{year}{1987}.

\bibitem[{Tannier et~al.(2007)Tannier, Bergeron, and Sagot}]{tannier-advances}
\bibinfo{author}{E.~Tannier}, \bibinfo{author}{A.~Bergeron},
  \bibinfo{author}{M.-F. Sagot}, \bibinfo{title}{Advances on sorting by
  reversals}, \bibinfo{journal}{Discrete Applied Mathematics}
  \bibinfo{volume}{155}~(\bibinfo{number}{6-7}) (\bibinfo{year}{2007})
  \bibinfo{pages}{881--888}.

\bibitem[{Tannier and Sagot(2004)}]{tannier-subquadratic}
\bibinfo{author}{E.~Tannier}, \bibinfo{author}{M.-F. Sagot},
  \bibinfo{title}{Sorting by Reversals in Subquadratic Time}, in:
  \bibinfo{booktitle}{Proceedings of the Fifteenth Annual Symposium on
  Combinatorial Pattern Matching (CPM)}, vol. \bibinfo{volume}{3109} of
  \emph{\bibinfo{series}{Lecture Notes in Computer Science}},
  \bibinfo{publisher}{Springer-Verlag}, ISBN \bibinfo{isbn}{3-540-22341-X},
  \bibinfo{pages}{1--13}, \bibinfo{year}{2004}.

\bibitem[{Bergeron et~al.(2002)Bergeron, Heber, and Stoye}]{bergeron-common}
\bibinfo{author}{A.~Bergeron}, \bibinfo{author}{S.~Heber},
  \bibinfo{author}{J.~Stoye}, \bibinfo{title}{Common intervals and sorting by
  reversals: a marriage of necessity}, \bibinfo{journal}{Bioinformatics}
  \bibinfo{volume}{18}~(\bibinfo{number}{Suppl 2}) (\bibinfo{year}{2002})
  \bibinfo{pages}{S54--S63}.

\bibitem[{Han(2006)}]{han-improving}
\bibinfo{author}{Y.~Han}, \bibinfo{title}{Improving the Efficiency of Sorting
  by Reversals}, in: \bibinfo{booktitle}{Proceedings of The 2006 International
  Conference on Bioinformatics and Computational Biology},
  \bibinfo{publisher}{CSREA Press}, \bibinfo{address}{Las Vegas, Nevada, USA},
  ISBN \bibinfo{isbn}{1-60132-002-7}, \bibinfo{pages}{4}, \bibinfo{year}{2006}.

\bibitem[{Bader et~al.(2001)Bader, Moret, and Yan}]{bader-linear}
\bibinfo{author}{D.~A. Bader}, \bibinfo{author}{B.~M.~E. Moret},
  \bibinfo{author}{M.~Yan}, \bibinfo{title}{A Linear-Time Algorithm for
  Computing Inversion Distance between Signed Permutations with an Experimental
  Study}, \bibinfo{journal}{Journal of Computational Biology}
  \bibinfo{volume}{8}~(\bibinfo{number}{5}) (\bibinfo{year}{2001})
  \bibinfo{pages}{483--491}.

\bibitem[{Gog and Bader(2008)}]{gog-fast}
\bibinfo{author}{S.~Gog}, \bibinfo{author}{M.~Bader}, \bibinfo{title}{Fast
  algorithms for transforming back and forth between a signed permutation and
  its equivalent simple permutation}, \bibinfo{journal}{Journal of
  Computational Biology} \bibinfo{volume}{15}~(\bibinfo{number}{8})
  (\bibinfo{year}{2008}) \bibinfo{pages}{1--13}.

\bibitem[{Hannenhalli and Pevzner(1996)}]{hannenhalli-tocut}
\bibinfo{author}{S.~Hannenhalli}, \bibinfo{author}{P.~A. Pevzner},
  \bibinfo{title}{To Cut... or Not to Cut (Applications of Comparative Physical
  Maps in Molecular Evolution)}, in: \bibinfo{booktitle}{Proceedings of the
  Seventh Annual ACM-SIAM Symposium on Discrete Algorithms (SODA)},
  \bibinfo{publisher}{ACM}, \bibinfo{address}{Atlanta, Georgia},
  \bibinfo{pages}{304--313}, \bibinfo{year}{1996}.

\bibitem[{Knuth(1995)}]{knuth-art}
\bibinfo{author}{D.~E. Knuth}, \bibinfo{title}{Sorting and Searching},
  vol.~\bibinfo{volume}{3} of \emph{\bibinfo{series}{The art of Computer
  Programming}}, \bibinfo{publisher}{Addison-Wesley}, ISBN
  \bibinfo{isbn}{0-201-03803-X}, \bibinfo{year}{1995}.

\bibitem[{Berman et~al.(2002)Berman, Hannenhalli, and
  Karpinski}]{berman-better}
\bibinfo{author}{P.~Berman}, \bibinfo{author}{S.~Hannenhalli},
  \bibinfo{author}{M.~Karpinski}, \bibinfo{title}{$1.375$-approximation
  algorithm for sorting by reversals}, in: \bibinfo{booktitle}{Proceedings of
  the Tenth Annual European Symposium on Algorithms (ESA)}, vol.
  \bibinfo{volume}{2461} of \emph{\bibinfo{series}{Lecture Notes in Computer
  Science}}, \bibinfo{publisher}{Springer-Verlag}, \bibinfo{address}{Rome,
  Italy}, \bibinfo{pages}{200--210}, \bibinfo{year}{2002}.

\end{thebibliography}

\end{document}